\documentclass[12pt, draftclsnofoot, onecolumn]{IEEEtran}
\usepackage{amsmath,amssymb,amsthm}

\usepackage{graphicx}
\usepackage[usenames]{color}
\usepackage[latin1]{inputenc}
\usepackage{dsfont}
\usepackage{cite}
\graphicspath{{figures/}}
\usepackage{array}
\usepackage{bm}
\usepackage{url}
\usepackage{color}
\usepackage{amsmath,amsthm}
\theoremstyle{remark}
\newtheorem{theorem}{Theorem}
\newtheorem{lemma}{Lemma}
\newtheorem{remark}{Remark}

\usepackage{enumerate}
\usepackage{amssymb}
\usepackage{graphicx}
\usepackage{subfigure}
\usepackage{epsfig}
\usepackage{psfrag}
\usepackage{verbatim}
\usepackage{algorithm}
\usepackage{algorithmic}
\usepackage{threeparttable}
\usepackage{cite}
\usepackage{multicol}
\usepackage{lipsum}

\usepackage{color}

\newcommand{\rmT}{\rm{T}}

\newcommand{\dof}{{\rm DoF}}

\newcommand{\rmH}{\rm{H}}

\newcommand{\rmvec}{{\rm{vec}}}

\newcommand{\SNR}{\rm{SNR}}

\newcommand{\bfH}{\mathbf{H}}
\newcommand{\bfA}{\mathbf{A}}
\newcommand{\bfa}{\mathbf{a}}
\newcommand{\bfC}{\mathbf{C}}
\newcommand{\bfc}{\mathbf{c}}

\newcommand{\bfI}{\mathbf{I}}

\newcommand{\bfR}{\mathbf{R}}

\newcommand{\bfw}{\mathbf{w}}
\newcommand{\bfW}{\mathbf{W}}

\newcommand{\mcD}{\mathcal{D}}

\newcommand{\mcC}{\mathcal{C}}

\newcommand{\logtwo}{\log_{2}}
\newcommand{\kdragger}{k^{\dagger}}
\newcommand{\lb}{\left(}
\newcommand{\rb}{\right)}
\newcommand{\xtrehat}{\hat{x}_{\rm th}}
\newcommand{\invk}{\frac{1}{K}}
\newcommand{\Pout}{P_{\rm out}}
\newcommand{\NR}{N_{\rm{R}}}
\newcommand{\NT}{N_{\rm{T}}}

\newcommand{\be}{\begin{equation}}
\newcommand{\ee}{\end{equation}}
\newcommand{\ba}{\begin{array}}
\newcommand{\ea}{\end{array}}
\newcommand{\bdm}{\begin{displaymath}}
\newcommand{\edm}{\end{displaymath}}
\newcommand{\bea}{\begin{eqnarray}}
\newcommand{\eea}{\end{eqnarray}}
\newcommand{\bean}{\begin{eqnarray*}}
\newcommand{\eean}{\end{eqnarray*}}
\newcommand{\xtre}{x_{\rm th}}

\newcommand{\dt}{{\rm d}x}
\makeatletter
\newcommand{\vast}{\bBigg@{3}}
\newcommand{\Vast}{\bBigg@{4}}
\makeatother
\def\[{\left[}
\def\]{\right]}

\begin{document}

\title{On the Degrees of Freedom for Opportunistic Interference Alignment with 1-Bit Feedback: The 3 Cell Case }
\author{
\IEEEauthorblockN{Zhinan Xu$^1$, Mingming Gan$^1$, Thomas Zemen$^{1,2}$}\\
\IEEEauthorblockA{$^1$ FTW (Telecommunications Research Center Vienna), Vienna, Austria\\
	$^2$AIT Austrian Institute of Technology, Vienna, Austria \\
Contact: xu@ftw.at}
}

\maketitle

\begin{abstract}
	Opportunistic interference alignment (OIA) exploits channel randomness and multiuser diversity by user selection. For OIA the transmitter needs channel state information (CSI), which is usually measured on the receiver side and sent to the transmitter side via a feedback channel. Lee and Choi show that $d$ degrees of freedom (DoF) per transmitter are achievable in a 3-cell MIMO interference channel assuming perfect real-valued feedback. However, the feedback of a real-valued variable still requires infinite rate.
	In this paper, we investigate 1-bit quantization for opportunistic interference alignment (OIA) in 3-cell interference channels. We prove that 1-bit feedback is sufficient to achieve the optimal DoF $d$ in 3-cell MIMO interference channels if the number of users per cell is scaled as ${\SNR}^{d^2}$. Importantly, the required number of users for OIA with 1-bit feedback remains the same as with real-valued feedback. For a given system configuration, we provide an optimal choice of the 1-bit quantizer, which captures most of the capacity provided by a system with real-valued feedback. Using our new 1-bit feedback scheme for OIA, we compare OIA with IA and show that OIA has a much lower complexity and provides a better rate in the practical operation region of a cellular communication system.

\end{abstract}

\section{Introduction}
Interference is a crucial limitation in next generation cellular systems. To address this problem, interference alignment (IA) has attracted much attention and has been extensively studied lately. IA is able to achieve the optimal degrees of freedom (DoF) at high signal-to-noise ratios (SNR) resulting in a rate of {$M/2\cdot {\rm{log}}({\rm{SNR}})+o({\rm{log}}({\rm{SNR}}))$ for the {$M$}} cell interference channel. For IA a closed-form solution of the precoding vectors for single antenna nodes with symbol extension is known \cite{Cadambe2008}. However, this coding scheme is based on the assumption that global channel state information (CSI) is available at all nodes, which is extremely hard to achieve and maybe even impossible. An iterative IA algorithm is proposed in \cite{Gomadam2011} to find the precoding matrices numerically with only local CSI at each node exploiting channel reciprocity. However, a number of iterations involving singular value decompositions (SVDs) have to be conducted which greatly increases the computational complexity. 

\subsection{Related Work}
For IA, CSI feedback has been investigated in \cite{Bolcskei2009,Krishnamachari2013,Xu2014a,Ayach2012a}. In \cite{Bolcskei2009}, channel coefficients are quantized using a Grassmannian codebook for frequency-selective single-input single-output (SISO) channels. The work in \cite{Krishnamachari2013} and \cite{Xu2014a} extends the results to multiple-input multiple-output (MIMO) channels and time-variant SISO channels respectively. The results in \cite{Bolcskei2009,Krishnamachari2013,Xu2014a} show that the full DoF is achievable as long as the feedback rate is high enough (which scales with the transmit power). Instead of quantizing the CSI, \cite{Ayach2012a} considers analog feedback and shows that the DoF of IA can be preserved as long as the forward and reverse link SNRs scale together. As the number of feedback bits increases, however, complexity increases and limited feedback becomes less practical due to undesirably large codebooks.

{For the sake of complexity reduction, opportunistic interference alignment (OIA) has been studied lately \cite{Yang2013,Jung2012,Gou2012,Lee2013,Yang2014,Lee2013b}. The key idea of OIA is to exploit the channel randomness and multiuser diversity by proper user selection. In \cite{Yang2013,Jung2012,Gou2012,Lee2013,Yang2014,Lee2013b}, signal subspace dimensions are used to align the interference signals. Each transmitter opportunistically selects and serves the user whose interference channels are most aligned to each other. The degree of alignment is quantified by a metric. To facilitate a user selection algorithm, all potential users associated with the transmitter are required to calculate and feedback the metric value based on the local CSI. Perfect IA can be achieved asymptotically if the number of users scales fast enough with SNR. The corresponding user scaling law to obtain the optimal DoF is characterized for multiple access channels in \cite{Jung2012,Yang2013} and for downlink interference channels in \cite{Lee2013,Yang2014,Lee2013b}. 
	
The work in \cite{Lee2013} decouples a multiple-input multiple-output (MIMO) interference channel into multiple SIMO interference channels and guarantees each selected user with one spatial stream. Since each stream is associated with one metric value, therefore multiple metric values have to be fed back from each user. The work of \cite{Yang2014} reduces the number of users to achieve the optimal DoF at the expense of increased feedback information from each user. In \cite{Yang2014}, each user has to feed back a metric value and a channel vector to cancel intra-cell interference. To enable multiple spatial streams for each selected user, the authors of \cite{Lee2013b} investigate the required user scaling in 3-cell MIMO interference channels and show that the optimal DoF $d$ is achieved if the number of users $K$ is scaled as $K \propto {\rm{SNR}}^{d^{2}}$. Therefore, at higher SNR, a larger number of users is required to achieve the optimal DoF. Clearly, the level of required total CSI feedback also increases proportionally to the number of users. However, in practical systems, the feedback is costly and the bandwidth of the feedback channel is limited. As a result, the feedback rate should be kept as small as possible.}
 
For opportunistic transmission in point-to-point systems, the problem of feedback reduction is tackled in \cite{Gesbert2004,Sharif2005,Sanayei2007a} by selective feedback. The solution is to let the users threshold their receive SNRs and notify the transmitter only if their SNR exceeds a predetermined threshold. The work in \cite{Gesbert2004,Sharif2005} reduces the number of real-valued variables that must be fed back to the transmitter in SISO and MIMO multiuser channels respectively. But \cite{Gesbert2004,Sharif2005} do not directly address the question of feedback rate since transmission of real-valued variables requires infinite rate. The work in \cite{Sanayei2007a} investigates the performance of opportunistic multiuser systems using limited feedback and proves that 1-bit feedback per user can capture a double-logarithmic capacity growth with the number of users. Note that \cite{Gesbert2004,Sharif2005,Sanayei2007a} consider {\em interference-free point-to-point transmissions.}

Unlike point-to-point systems where the imperfect CSI causes only an SNR offset in the capacity, the accuracy of the CSI in interference channels affects {\em the slope of the rate curve}, i.e., the DoF. Thus, for OIA, a relation to the DoF using selective feedback is critical. Can we reduce the amount of feedback and still preserve the optimal DoF? This is addressed in our paper \cite{Xu2015} using real-valued feedback. It shows that the amount of feedback can be dramatically reduced by more than one order of magnitude while still preserving the essential DoF promised by conventional OIA with perfect real-valued feedback. However, to the best of our knowledge, the achievability of the optimal DoF with limited feedback is still unknown.\footnote[1] {We are interested in limited feedback for the metric value. The work of \cite{Yang2014} addresses limited feedback to quantize a channel vector, which is not relevant to our work.} Our previous work \cite{Xu2015a} tackles this problem by 1-bit feedback to achieve the DoF $d=1$. This paper generalizes the results of \cite{Xu2015a} also to the cases of $d>1$.

\subsection{Contributions}
In this paper, we consider 1-bit feedback for 3-cell MIMO interference channels. 
\begin{itemize}
\item We prove that only 1-bit feedback per user is sufficient to achieve the full DoF (without requiring more users than real-valued feedback) if the one-bit quantizer is chosen judiciously.
\item We derive the scheduling outage probability according to the metric distribution for 1-bit feedback. 
\item We provide an optimal choice of the 1-bit quantizer to achieve the DoF of 1, which captures most of the capacity provided by a system with real-valued feedback. To achieve a DoF $d>1$, an asymptotic threshold choice is given by solving an upper bound for the rate loss.
\item The DoF achievable threshold is not unique. We generalize the design of the threshold choices and provide the mathematical expression.
\item We compare OIA and IA with the same amount of feedback and present the comparison in terms of complexity and achievable rate. We show that OIA has a much simpler quantizer and provides a higher sum rate in the practical operation region of a cellular communication system.
\end{itemize}

\subsection{Organization}
The rest of the paper is organized as follows. In Section \ref{sec2}, we introduce the system model of OIA. Section \ref{sec3} provides the background, the achievable DoF and user scaling law for conventional OIA. Section \ref{sec4} describes the proposed 1-bit feedback scheme and derives the optimal and asymptotic optimal choices for the 1-bit quantizer. The numerical results are provided in Section \ref{sec5}. In Section \ref{sec6}, we give a comprehensive comparison between IA with limited feedback and OIA with 1-bit feedback. Finally, 
we conclude the paper in Section \ref{sec7}.

\subsection{Notations}
We denote a scalar by $a$, a column vector by $\bfa$ and a matrix by $\bfA$. The superscript $^{\rmT}$ and $^{\rmH}$ stand for transpose and Hermitian transpose, respectively. The notations $\left\| \cdot \right\|$, $\left\| \cdot \right\|_{\rm F}$, $\rmvec(\cdot)$, $\det(\cdot)$, $\lceil{\cdot}\rceil$ and $\mathbb{E}[\cdot]$ denote vector 2-norm, Frobenius norm, vecterization, determinant, ceiling operation and the expectation operation, respectively. $\bfI_N$ is the $N \times N$ identity matrix. For a given function $f(N)$, we write $g(N)=O(f(N))$ if and only if $\lim_{N \rightarrow \infty}|g(N)/f(N)|$ is bounded and $g(N)=o(f(N))$ if and only if $\lim_{N \rightarrow \infty}|g(N)/f(N)|=0$. 
$\log$ is the natural logarithm function. 
\section{System model} \label{sec2}
Let us consider the system model for the $3$-cell MIMO interference channel, as shown in Fig.~\ref{MU_MIMO}. It consists of 3 transmitters with $N_{\rm{T}}$ antennas, each serving $K$ users with $N_{\rm{R}}$ antennas. The channel matrix from transmitter $j$ to receiver $k$ in cell $i$ is denoted by $\mathbf{H}_{i,j}^{k} \in \mathbb{C}^{N_{\rm{R}} \times N_{\rm{T}}}$, $\forall i,j \in \{ 1,2,3\}$ and $k \in \{ 1,\ldots,K\}$. Every element of $\mathbf{H}_{i,j}^{k}$ is assumed as an independent identically distributed (i.i.d.) symmetric complex Gaussian random variable with zero mean and unit variance.

For a given transmitter, its signal is only intended to be received and decoded by a single user for a given signaling interval. The signal received at receiver $k\in\{1,\ldots, K\}$ in cell $i$ at a given time instant is the superposition of the signals transmitted by all three transmitters, which can be written as
\begin{align}
\mathbf{x}_{i}^{k}=\mathbf{H}_{i,i}^{k}\mathbf{s}_i+{\sum_{j=1,j{\neq}i}^{3}\mathbf{H}_{i,j}^{k}\mathbf{s}_j}+\mathbf{n}_{i}^{k}, \label{sysmodel}
\end{align}
where vector $\mathbf{s}_j \in \mathbb{C}^{d \times 1}$ denotes $d$ transmitted symbols from transmitter $j$ with power constraint $\mathbb{E}\{\mathbf{s}_j\mathbf{s}_j^{\rm{H}}\}=\frac{P}{d}\mathbf{I}_d$.  The additive complex symmetric Gaussian noise $\mathbf{n}_i^{k} \sim \mathcal{C} \mathcal{N}(0, \mathbf{I}_{\NR})$ has zero mean and unit variance. Thus, the SNR becomes ${\rm SNR}=P$. In this paper, we confine ourselves to the case of $N_{\rm{R}}=2d$ and $N_{\rm{T}} = d$. This is interesting because it is the minimum setup  to achieve the full DoF $d$ at each receiver. 
In case the number of receive antennas $N_{\rm{R}}>2d$, $N_{\rm{R}}-2d$ DoF can be obtained with probability one even without interference management because uncoordinated interference signals will span a subspace with a maximum of $2d$ dimensions in the space $\mathbb{C}^{N_{\rm{R}}}$. On the other hand if $N_{\rm{R}}<2d$, the full DoF $d$ is not achievable because the interference signals will span at least a $d$ dimensional subspace even when they are perfectly aligned. The model in (\ref{sysmodel}) is statistically equivalent to the case when $N_{\rm{T}} \geq d$ and a linear precoding matrix $\mathbf{V}_j \in \mathbb{C}^{{N}_{\rm{T}} \times d}$ is applied to each transmitter as $\mathbf{x}_{i}^{k}=\mathbf{H}_{i,i}^{k}\mathbf{V}_i\mathbf{s}_i+{\sum_{j=1,j{\neq}i}^{3}\mathbf{H}_{i,j}^{k}\mathbf{V}_j\mathbf{s}_j}+\mathbf{n}_i^{k}$.

\begin{figure}[t]
	\centering
	\includegraphics[width=0.6\columnwidth]{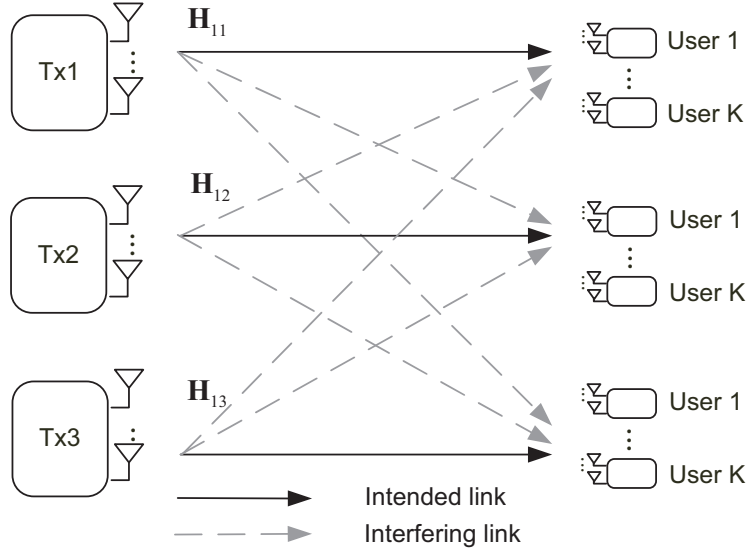}
	\centering
	\caption{Three-cell MIMO interference channel with $K$ candidates in each cell}
	\label{MU_MIMO}
\end{figure}

Defining ${\mathbf{U}_i^{k}} \in \mathbb{C}^{N_{\rm{R}} \times d}$ as the postfiltering matrix at receiver $k$ in cell $i$, the received signal of user $k$ in cell $i$ becomes
\begin{align}
\mathbf{y}_{i}^{k}=&{\mathbf{U}_i^{k}}^{\rm{H}} \mathbf{x}_{i}^{k}\nonumber\\
=&{\mathbf{U}_i^{k}}^{\rm{H}}{\mathbf{H}}_{i,i}^{k}\mathbf{s}_i+\sum_{j=1,j{\neq}i}^{3}{\mathbf{U}_i^{k}}^{\rm{H}}{\mathbf{H}}_{i,j}^{k}\mathbf{s}_j+{\bar{\mathbf{n}}}_{i}^{k}
\end{align}
where ${\bar{\mathbf{n}}}_{i}^{k}={\mathbf{U}_{i}^{k}}^{\rm{H}}\mathbf{n}_i^{k}$ denotes the effective spatially white noise vector. The achievable instantaneous rate for user $k$ in cell $i$ becomes

\begin{align}
R_{i}^{k}=& {\rm{log}}_2 {\rm{det}} \bigg( \mathbf{I}_{d}+\frac{P}{d} {\mathbf{U}_i^{k}}^{\rm{H}} {\mathbf{H}}_{i,i}^{k} 
{{{\mathbf{H}}}_{i,i}^{k}}^{\rmH}{\mathbf{U}_i^{k}} \Big( \frac{P}{d} \sum_{j=1,j{\neq}i}^{3} {\mathbf{U}_i^{k }}^{\rmH}{\mathbf{H}}_{i,j}^{k} 
{{{\mathbf{H}}}_{i,j}^{k}}^{\rmH}{\mathbf{U}_i^{k}} +\mathbf{I}_{d}\Big)^{-1}\bigg) \\
=&\underbrace{{\rm{log}}_2 {\rm{det}} \lb \mathbf{I}_{d}+\sum_{j=1}^{3}\frac{P}{d} {\mathbf{U}_i^{k}}^{\rm{H}} {\mathbf{H}}_{i,i}^{k} 
	{{{\mathbf{H}}}_{i,i}^{k}}^{\rmH}{\mathbf{U}_i^{k}} \rb}_{{R_{\rm gain}}_{i}^{k}} - \underbrace{{\rm{log}}_2 {\rm{det}} \lb \mathbf{I}_{d}+\sum_{j=1,j{\neq}i}^{3}\frac{P}{d} {\mathbf{U}_i^{k}}^{\rm{H}} {\mathbf{H}}_{i,i}^{k} 
	{{{\mathbf{H}}}_{i,i}^{k}}^{\rmH}{\mathbf{U}_i^{k}} \rb }_{{R_{\rm loss}}_{i}^{k}} \label{rategainloss}
\end{align}
where in (\ref{rategainloss}) we decompose the achievable rate into a rate gain term ${R_{\rm gain}}_{i}^{k}$ and a rate loss term ${R_{\rm loss}}_{i}^{k}$. Therefore, the DoF achieved for user $k$ in cell $i$ can be written as
\begin{align}
\dof_{i}^{k}=&\lim_{P\rightarrow\infty}  \frac{ {\mathbb E}[R_{i}^{k}]}{\logtwo P}\\
=& d- \underbrace{ \lim_{P\rightarrow\infty} \frac{{\mathbb E}[{R_{\rm loss}}_{i}^{k}]}{\logtwo P} }_{{\dof_{\rm loss}}_{i}^{k}}\label{dofeq}
\end{align}
where (\ref{dofeq}) is obtained due to $\lim_{P\rightarrow\infty} \frac{{\mathbb E}[{R_{\rm gain}}_{i}^{k}]}{\logtwo P}=d$. Therefore, in the rest of the paper, we will focus on the rate loss and DoF loss terms in order to analyze the achieved DoF.

\section{Conventional OIA} \label{sec3}
Without requiring global channel knowledge, OIA is able to achieve the same DoF as IA with only local CSI feedback within a cell. In this section, we describe the selection criteria and the design of the postfilter for the conventional OIA algorithm. The key idea of OIA \cite{Lee2013b} is to exploit the channel randomness and the multi-user diversity, using the following procedure: 
\begin{itemize}
	\item Each transmitter sends out a reference signal. 
	\item Each user equipment measures the channel quality using a specific metric.
	\item Every user feeds back the value of the metric to its own transmitter. 
	\item The transmitter selects a user in its own cell for communication according to the feedback values. 
\end{itemize}
We denote the index of the selected user in cell $i$ by $k^{\ast}$.
The transmitters aim at choosing a user, who observes most aligned interference signals from the other transmitters. The degree of alignment is quantified by a subspace distance measure, named chordal distance. It is generally defined as
\begin{eqnarray}
d_{\rm c}(\mathbf{A},\mathbf{B})=&1/\sqrt{2}\left\| \mathbf{A}\mathbf{A}^{\rm H}-\mathbf{B}\mathbf{B}^{\rm H} \right\|_{\rm F} \label{defdc}
\end{eqnarray}
where $\mathbf{A}$, $\mathbf{B} \in \mathbb{C}^{N_{\rm{R}} \times d}$ are the orthonormal bases of two subspaces and ${d_{\rm c}}^2(\mathbf{A}, \mathbf{B}) \leq d$. 
For OIA, each user finds an orthonormal basis $\mathbf{Q}$ of the column space spanned by the two interference channels respectively, i.e.,
${{{\mathbf{Q}}}_{ip}^{k}} \in {\rm span}({{\mathbf{H}}}_{ip}^{k})$ and ${{{\mathbf{Q}}}_{iq}^{k}} \in {\rm span}({{{\mathbf{H}}}_{iq}^{k}})$ where $p=(i+1 \mod 3)$ and $q=(i+2 \mod 3)$. Then the users calculate the distance between two interference subspaces using the obtained orthonormal basis, yielding 
\begin{align}
\mcD_{i}^{k}=d_{\rm c}^2({{{\mathbf{Q}}}_{ip}^{k}},{{{\mathbf{Q}}}_{iq}^{k}}),
\end{align}
where $\mcD_{i}^{k}$ is the distance measured at user $k$ in cell $i$. For conventional OIA, all users feed back the distance measure to their own transmitter and the user selected by transmitter $i$ is given by 
\begin{align}
k^{\ast}={\rm{arg~}} \underset{k}{\min ~} \mcD_{i}^{k}.
\end{align}
Therefore, the metric value of the selected user becomes $\mcD_{i}^{k^{\ast}}$. Defining the received interference covariance matrix of the selected user $k^{\ast}$ as
\begin{align}
\bfR_i^{k^{\ast}}= {\bfH}_{ip}^{k^{\ast}} {{\bfH}_{ip}^{k^{\ast}}}^{\rmH}+{\bfH}_{iq}^{k^{\ast}} {{\bfH}_{iq}^{k^{\ast}}}^{\rmH},
\end{align}
the postfilter applied at the selected user becomes
\begin{align}
\mathbf{U}_i^{k^{\ast}}=[\vec{\mathbf{u}}_{d+1}(\mathbf{R}_i^{k^{\ast}}),\cdots, \vec{\mathbf{u}}_{N_{\rm{R}}}(\mathbf{R}_i^{k^{\ast}})]
\end{align}
where $\vec{\mathbf{u}}_{n}(\bfR)$ represent the singular vector corresponding to the $n$-th largest singular value of $\bfR$.

\subsection{Achievable DoF of Conventional OIA}\label{TOIA}
As shown in \cite{Dai2008}, for quantizing a source $\bfA$ arbitrarily distributed on the Grassmannian manifold $\mathcal{G}_{\NR,d}(\mathbb{C})$ by using a random codebook $\mcC_{\rm rnd}$ with $K$ codewords, the second moment of the chordal distance can be bounded as
\begin{align}
Q(K)&=\mathbb{E} \[\underset{\bfC_{k}\in\mcC_{\rm rnd}}{\rm{min~}} d_{\rm c}^2(\bfA,\bfC_{k})\]&  \\
&\leq \frac{\Gamma(\frac{1}{d(\NR-d)})}{{d(\NR-d)}}{(K c_{\NR,d} )}^{-\frac{1}{{d(\NR-d)}}}
\end{align}
where $\Gamma(\cdot)$ denotes the Gamma function and the random codebook $\mcC_{\rm rnd}\subset\mathcal{G}_{\NR,d}(\mathbb{C})$. The constant $c_{\NR,d}$ is the ball volume on the Grassmannian manifold $\mathcal{G}_{\NR,d}(\mathbb{C})$, i.e.~
\begin{align}
c_{\NR,d}=\frac{1}{\Gamma(d(\NR-d)+1)}\prod_{i=1}^{d}\frac{\Gamma(\NR-i+1)}{\Gamma(d-i+1)}.
\end{align}
The problem of selecting the best user out of $K$ users is equivalent to quantizing an arbitrary subspace with $K$ random subspaces on the Grassmannian manifold $\mathcal{G}_{\NR,d}(\mathbb{C})$ \cite[Lemma 4]{Lee2013b}. Therefore, we have
$\mathbb{E} \[\mcD_{i}^{k^{}} \]=Q(1)$ and
$\mathbb{E} \[\mcD_{i}^{k^{\ast}} \]=Q(K)$.

We briefly revisit the results obtained in \cite{Lee2013b}, which will be used for comparison with our 1-bit feedback OIA.
A finite number of users $K$ results in residual interference. 
When the cell $i$ has $K$ users, the average rate loss at the selected user $k^{\ast}$ can be bounded as
\begin{align}
{\mathbb E}[{R_{\rm loss}}_{i}^{k^{\ast}}] \leq& d \cdot{\rm{log}}_2 \left(1+\frac{P}{d} \cdot \mathbb{E} [\mcD_{i}^{k^{\ast}}] \right) \label{rloss}\\
=& d \cdot{\rm{log}}_2 \left(1+\frac{P}{d} \cdot Q(K )\right)\label{rloss1},
\end{align}
where (\ref{rloss}) is obtained due to \cite[Lemma 6]{Lee2013b}.

The achievable DoF of transmitter $i$ using OIA can be expressed by $d-\lim_{P\rightarrow\infty} \frac{{\mathbb E}[{R_{\rm loss}}_{i}^{k^{\ast}}]}{\log_2 P}$. In order to achieve the DoF of $d'$, the number of users per cell has to be scaled as \cite[Theorem 2]{Lee2013b}
\begin{align}
K \propto P^{dd'}.
\end{align}

\section{The Achievable DoF of OIA with 1-Bit Feedback} \label{sec4}
In this section, we introduce the concept of 1-bit feedback for OIA. The achievability of the DoF is proven for $d=1$ first, where a closed-form solution exists. We generalize the result to all $d > 1$ based on asymptotic analysis.
\subsection{One-Bit Feedback by Thresholding}\label{SFT}
For conventional OIA, the user selected for transmission is the one with the smallest chordal distance measure. This requires that the transmitter collects the perfect real-valued chordal distance measures from all the users. However, the feedback of real values require infinite bandwidth. The question of how to efficiently feedback the required CSI is still not solved for OIA. To address this problem, we propose a threshold-based 1-bit feedback strategy where each user compares the locally measured chordal distance to a predefined threshold $\xtre$ and reports 1-bit information to the transmitter about the comparison. In such a way, the transmitter can partition all the users into two groups and schedule a user from the favorable group for transmission. Therefore, we propose the following steps for OIA using 1-bit feedback: 
\begin{itemize}
	\item Each transmitter sends out a reference signal. 
	\item Each user equipment measures the channel quality using the chordal distance measure. 
	\item Each user compares the locally measured chordal distance to a threshold. In case the measured value is smaller than the threshold, a '1' will be fed back; otherwise a '0' will be fed back. 
	\item The transmitter will randomly select a random user whose feedback value is '1' for transmission. 
\end{itemize}

A scheduling outage occurs if all users send '0' to the transmitter. In such an event, a random user among all users will be selected for transmission. To find the scheduling outage probability $\Pout$, we first denote the cumulative density function (CDF) of $\mcD_{i}^{k}$ by $F_\mcD(x)$, which is defined as
\begin{align}
F_\mcD(x)&={\rm Pr}(\mcD_{i}^{k} \leq x)\\
&={\rm Pr}(d_{\rm c}^2(\bfA,\bfC_{k}) \leq x)\\
&\approx
\begin{cases}
0, &  x<0\\
c_{\NR,d} \cdot x^{d(\NR-d)}, &  0\leq x \leq \hat{x} \label{cdff}\\
1, &  x>\hat{x}
\end{cases}
\end{align}
where $\hat{x}$ satisfies $c_{\NR,d} \cdot {\hat{x}}^{d(\NR-d)}=1$ and $\hat{x} \leq d$. If $d=1$, the CDF of (\ref{cdff}) becomes exact. If $d>1$, the CDF in (\ref{cdff}) is exact when $0 \leq x \leq 1$. When $1 < x < d$, the CDF provided by (\ref{cdff}) deviates from the true CDF \cite{Dai2008}. However, we are mainly interested in small $x <1 $ for the purpose of feedback reduction by thresholding. 

Therefore, the scheduling outage probability corresponds to the event where all $K$ users exceed $x$, which is denoted by
\begin{align} 
\Pout&={\rm Pr}(\underset{k}{\min ~} \mcD_{i}^{k} \geq x)\\
&={\rm Pr}(\underset{\bfC_{k}\in\mcC_{\rm rnd}}{\rm{min~}} d_{\rm c}^2(\bfA,\bfC_{k})\ \geq x)\\
&=\left(1-F_\mcD\left(\xtre\right)\right)^K.
\end{align}

We define the probability density functions (PDFs) of $\mcD_{i}^{k}$ as $f_\mcD(x)$, where $\int_{0}^{x} f_\mcD(x)\dt=F_\mcD(x)$.
In order to distinguish from the previous conventional OIA, we employ $k^{\dagger}$ as the index of the selected user with 1-bit feedback. The expected metric value of the selected user $k^{\dagger}$ can be expressed as

\begin{align}
	\mathbb{E} [\mcD_{i}^{k^{\dagger}}] =\left(1-\Pout\right) \int_{0}^{\xtre} \frac{ f_\mcD(x)x}{F_\mcD(\xtre)}\dt + \Pout \int_{\xtre}^{d} \frac{ f_\mcD(x) x}{1-F_\mcD(\xtre)} \dt,\label{avrd0}
\end{align}
where $\frac{ f_\mcD(x)}{F_\mcD(\xtre)}$ and $\frac{ f_\mcD(x)}{1-F_\mcD(\xtre)}$ are the normalized truncated PDFs of $\mcD_{i}^{k}$ in the corresponding intervals $[0,\xtre)$ and $[\xtre,d]$, satisfying

\begin{align}
	\int_{0}^{\xtre} \frac{ f_\mcD(x) \dt}{F_\mcD(\xtre)}=1 \quad\text{and} \quad \int_{\xtre}^{d}\frac{ f_\mcD(x) \dt}{1-F_\mcD(\xtre)}=1. 
\end{align}
The first term in (\ref{avrd0}) represents the event where at least one user falls below the threshold and reports '1' to the transmitter. The second term denotes a scheduling outage, where all the users exceed the threshold and report '0'.
 
\subsection{Achievable DoF and User Scaling Law When $d=1$} \label{deq1}
For a given $K$, $\Pout$ is uniquely determined by the choice of the threshold $\xtre$. We intend to find the optimal $\xtre$, such that (\ref{avrd0}) is minimized. The function is convex in the range of $[0,1]$. Thus, $\mathbb{E} [\mcD_{i}^{k^{\dagger}}]$ has an unique minimum within the interval $[0,1]$. To find the minimum value and the corresponding threshold, we need to solve the equation $\frac{\partial \mathbb{E} [\mcD_{i}^{k^{\dagger}}]}{\partial \xtre}=0$. For $d=1$, according to (\ref{cdff}) we have $F_\mcD(x)=x$ and $f_\mcD(x)=1$ in the interval $[0,1]$. The expected metric value $\mathbb{E} [\mcD_{i}^{k^{\dagger}}]$ in (\ref{avrd0}) can be simplified as
\begin{align}
D_i(\xtre)&=\mathbb{E}[\mcD_{i}^{k^{\dagger}}] \nonumber\\
&=\left(1-\Pout\right) \int_{0}^{\xtre} \frac{ x\dt}{\xtre} + \Pout \frac{\int_{\xtre}^{1} x \dt}{1-\xtre}\label{avrd} \nonumber\\
&=(1-(1-\xtre)^K) \frac{\xtre}{2}  +(1-\xtre)^K(\frac{1+\xtre}{2}).
\end{align}
The optimal $\xtre$ which minimizes $\mathbb{E}[\mcD_{i}^{k^{\dagger}}]$ can be found by solving $\frac{\partial D_i(\xtre)}{\partial \xtre}=0$, i.e.~$-K(1-\xtre)^{K-1}+1=0$.
Thus we have the optimal threshold 
\begin{align}
\xtrehat=1-(\invk)^{\frac{1}{K-1}} \label{opttre}. 
\end{align}
Applying $\xtrehat$ to (\ref{avrd}), the minimum of $D_i(\xtre)$ can be written as a function of $K$ as
\begin{align}
D_i(\xtrehat)=\frac{1}{2}\left(\invk \right)^{\frac{K}{K-1}}-\frac{1}{2}\left(\invk\right)^{\frac{1}{K-1}}+\frac{1}{2} \label{dmin}.
\end{align}
This leads us to the following lemma, which will then be used for the proof of the achievable DoF.
\begin{lemma} \label{eqpairlemma}
When the number of users $K$ goes to infinity, i.e.~$K \rightarrow \infty$, $D_i(\xtrehat)$ is asymptotically equivalent to $\frac{\log(K)}{2K}$, such that
\begin{align}
\lim_{K \rightarrow \infty} \frac{D_i(\xtrehat)}{\frac{\log K}{2K}}=1 \label{lem1}.
\end{align}
\end{lemma}

\begin{proof}
Accroding to (\ref{dmin}), the left hand side of (\ref{lem1}) can be written as
\begin{align}
&\lim_{K \rightarrow \infty} \frac{  \left(\frac{1}{K} \right)^{\frac{K}{K-1}}-\left(\frac{1}{K}\right)^{\frac{1}{K-1}}+1}{\frac{\log K}{K}} \\
&=\lim_{K \rightarrow \infty} \frac{  \left(\frac{1}{K} \right)-\left(\frac{1}{K}\right)^{\frac{1}{K}}+1}{\frac{\log K}{K}} \label{KtoM} \\
&=\lim_{M \rightarrow 0} \frac{M^M(\log M+1)-1}{\log M+1} \label{LBD}\\
&=\lim_{M \rightarrow 0} M^M - \lim_{M \rightarrow 0} \frac{1}{\log M+1}\\
&=1 \nonumber
\end{align}
where (\ref{LBD}) is obtained by letting $M=1/K$ and applying the L'Hôpital's rule. Thus, the proof is complete.
\end{proof}


\begin{theorem} \label{themdeq1}
For $d=1$, if the number of users is scaled as $K \propto P^{d'}$, 1-bit feedback per user is able to achieve a DoF $d' \in [0,1]$ per transmitter if the the threshold is optimally chosen according to (\ref{opttre}).
\end{theorem}

\begin{proof}
The achievable DoF of transmitter $i$ using OIA can be expressed as $1-d_{\rm loss}$. If $K \propto P^{d'}$, the DoF loss term can be written as 
\begin{align}
d_{\rm loss}&=\lim_{P\rightarrow\infty} \frac{{\mathbb E}[{R_{\rm loss}}_{i}^{k^{\dagger}}]}{\log_2 P} \\
&\leq\lim_{P\rightarrow\infty} \frac{{\log}_2 \left(1+P D_i\left(\xtrehat\right)\right)}{\log_2 P} \label{dofupper}\\
&=\lim_{P\rightarrow\infty} \frac{{\log}_2 \left(P D_i\left(\xtrehat\right)\right)}{{\log}_2 P} \\
&=\lim_{P\rightarrow\infty} \frac{{\log}_2 \left(P \cdot \frac{\log K}{2K}\right)}{{\log}_2 P} \label{eqpair}\\
&=(1-d')+\lim_{P\rightarrow\infty} \frac{1}{{\log} P + O(1)} \label{eqpair1}\\
&=(1-d').
\end{align}
The inequality (\ref{dofupper}) is obtained by using the upper bound in (\ref{rloss}) and invoking (\ref{dmin}). Equality (\ref{eqpair}) is due to the asymptotic equivalence in Lemma \ref{eqpairlemma}. Equality (\ref{eqpair1}) is obtained using the relationship $K \propto P^{d'}$ and the L'Hôpital's rule. Therefore, the DoF $d'$ is obtained at each transmitter.
\end{proof}

\begin{remark}
Compared to conventional OIA in \cite{Lee2013b}, the user scaling law achieving DoF $d'$ remains the same. The second term in (\ref{eqpair1}) does not exist for conventional OIA. However, it goes to $0$ when $P \rightarrow \infty$, and thus does not change the DoF. Therefore, 1-bit feedback neither degrades the performance in terms of DoF nor requires more users to achieve the same DoF.
\end{remark}

\subsection{Achievable DoF and User Scaling Law When $d>1$} \label{dl1}
 Now we want to generalize the result to any $d$ values. However, for $d > 1$, a closed-form solution does not exist. In this section, we will base our investigation on asymptotic analysis. To ease the notation, we drop the dependence of $c_{\NR,d}$ on $d$ and let $\NR=2d$. First, we simplify (\ref{avrd0}) using the following upper bound
	\begin{align}
		&\mathbb{E} [\mcD_{i}^{k^{\dagger}}] \nonumber\\ 
		&=\left(1-\Pout\right) \int_{0}^{\xtre} \frac{ f_\mcD(x)x}{F_\mcD(\xtre)}\dt + \Pout \int_{\xtre}^{d} \frac{ f_\mcD(x) x}{1-F_\mcD(\xtre)} \dt \nonumber\\
		&\leq \left(1-\Pout\right) \xtre + \Pout d \label{upperbinterg}\\
		&= \xtre+(d-\xtre)(1-F_\mcD(\xtre))^K \\
		&= \xtre+(d-\xtre)(1-c{\xtre}^{d^2})^K \label{upeq41}
	\end{align}
	where (\ref{upperbinterg}) is obtained by taking the upper limit of the integration. To find the minimum value and the corresponding threshold, we need to solve the partial derivative of (\ref{upeq41}) with respect to $\xtre$, i.e.

	\begin{align}
		1-(1-c{\xtre}^{d^2})^K-cKd^{2}(d-\xtre){\xtre}^{d^{2}-1}(1-c{\xtre}^{d^2})^{K-1} =0.
	\end{align}
	where an explicit solution does not exist for $d>1$ to the best of our knowledge. 
	
	Therefore, instead of an explicit solution, we will find an asymptotically close solution. We simplify equation (\ref{upeq41}) by letting $y=c{\xtre}^{d^2}$, i.e.
	\begin{align}
		\mathbb{E} [\mcD_{i}^{k^{\dagger}}] &\leq \xtre+(d-\xtre)(1-c{\xtre}^{d^2})^K \nonumber \\
		&=\lb\frac{y}{c}\rb^{\frac{1}{d^2}}+\lb  d-\lb\frac{y}{c}\rb ^{{\frac{1}{d^2}}} \rb(1-y)^K \\
		&\leq(\frac{y}{c})^{\frac{1}{d^2}}+d \sum_{n=0}^{\infty}(-1)^n\binom{K}{n}y^n \label{tylor1}
	\end{align}
	where (\ref{tylor1}) is obtained by neglecting $(\frac{y}{c})^{{\frac{1}{d^2}}}$ in the second term and applying the Maclaurin series expansion to the following binomial function
	\begin{align}
		&(1-y)^K\nonumber\\
		&=1-Ky+\frac{K(K-1)y^2}{2!}\cdots+(-1)^{n}\frac{K\cdots(K-n+1)y^n}{n!}\nonumber\\
		&=\sum_{n=0}^{\infty}(-1)^n\binom{K}{n}y^n.
	\end{align}
	To proceed our proof, we give the following lemma.
	\begin{lemma} \label{eqpairlemmadl1}
		When the number of users $K$ goes to infinity, i.e.~$K \rightarrow \infty$, the binomial coefficient 
		\begin{align}
			\binom{K}{n}=\frac{K^n}{n!}\lb1+O\lb\frac{1}{K}\rb\rb.
		\end{align}
	\end{lemma}
	\begin{proof}
		By definition of $\binom{K}{n}$, we have
		\begin{align}
			\binom{K}{n} &= \frac{K!}{n!(K-n)!} \nonumber\\
			&=\frac{(K-n+1)(K-n+2)\cdots K}{n!} \label{binom1}
		\end{align}
		The numerator in (\ref{binom1}) can be expanded as
		\begin{align}
			&(K-n+1)(K-n-1)...K \nonumber\\
			&= K^n + c_1(n)K^{n-1} + c_2(n)K^{n-2} + \cdots + c_n(n) 
		\end{align}
		where $c_i(n)$ are polynomial functions dependent only on $K$. When $K \rightarrow \infty$, we can extract $K^n$ to obtain
		\begin{align}
			K^n(1 + \frac{c_1(n)}{K} + \frac{c_2(n)}{K^2} + \cdots + \frac{c_n(n)}{K^n}) =K^n\lb 1+O \lb\frac{1}{K}\rb \rb \nonumber
		\end{align}
		and thus $\binom{K}{n} = \frac{K^n}{n!}\lb 1+O\lb\frac{1}{K} \rb \rb$.
	\end{proof}
	Therefore, when $K \rightarrow \infty$, (\ref{tylor1}) can be written as
	\begin{align}
		\mathbb{E} [\mcD_{i}^{k^{\dagger}}] &\leq\lb\frac{y}{c}\rb^{\frac{1}{d^2}}+d \sum_{n=0}^{\infty}(-1)^n\binom{K}{n}y^n \nonumber\\
		&=\lb\frac{y}{c}\rb^{\frac{1}{d^2}}+d \lb 1+O\lb \frac{1}{K} \rb\rb \sum_{n=0}^{\infty}(-1)^n \frac{K^n y^n}{n!}  \label{eqlemma1}\\
		&=\lb\frac{y}{c}\rb^{\frac{1}{d^2}}+d\lb 1+O\lb \frac{1}{K} \rb\rb e^{-Ky}  \label{eqlemmaexp} \\
		&= \underbrace{ \lb\frac{y}{c}\rb^{\frac{1}{d^2}}+d e^{-Ky}}_{{\tilde{D}_{i} (y)}} \label{ublemmaexp1}
	\end{align}
	where (\ref{eqlemma1}) follows from lemma \ref{eqpairlemmadl1}. Equality (\ref{eqlemmaexp}) is obtained by utilizing the Maclaurin series expansion of the exponential function
	\begin{align}
		e^{-Ky}&=1-Ky+\frac{K^2 y^2}{2!}-\frac{K^3 y^3}{3!}+\cdots+(-1)^n \frac{K^n y^n}{n!} \nonumber\\
		&=\sum_{n=0}^{\infty}(-1)^n \frac{K^n y^n}{n!}.
	\end{align}
	Equality (\ref{ublemmaexp1}) is obtained by neglecting $O\lb \frac{1}{K} \rb$ due to the fact $K \rightarrow \infty$. 
	We define ${\tilde{D}_{i} (y)}$ as the upper bound obtained in (\ref{ublemmaexp1}). The $y$ which minimizes ${\tilde{D}_{i} (y)}$ is the solution to  
	\begin{align}
		\frac{\partial {\tilde{D}_{i} (y)}}{\partial y}=\frac{1}{d^2}\lb {\frac{y}{c}}\rb^{\lb \frac{1}{d^2}-1 \rb} - d K e^{-Ky}=0 \label{deriv1}. 
	\end{align}
	For ({\ref{deriv1}}), the real solutions should exist in $(0,\infty)$, which can be found by numerical approximation. However, for general $d$ (expect for $d=1$), an explicit solution is still mathematically intractable. The solver can be written in the form of the Lambert W function \cite{Corless1996}, which is a set of functions satisfying $W(z)e^{W(z)}=z$. To this end, we first rewrite (\ref{deriv1}) as
	\begin{align}
		\frac{K}{\alpha}y e^{\frac{K}{\alpha}y}=\frac{Kc\lb d^3K\rb^{\frac{1}{\alpha}}}{\alpha}
	\end{align}
	where $\alpha=\frac{1}{d^2}-1$. The possible real solutions to this equation are given by
	\begin{align}
		\hat{y}=\frac{\alpha \cdot W_{\zeta} \lb \frac{ Kc\lb d^3K\rb^{\frac{1}{\alpha}}}{\alpha}\rb}{K}, \zeta\in \{0,-1\} \label{solutionlambert},
	\end{align}
	where the function $W_{0}(\cdot)$ and $W_{-1}(\cdot)$ are two real branches of the Lambert W function defined in the intervals $[-\frac{1}{e},\infty)$ and $[-\frac{1}{e},0)$, corresponding to the maximum and minimum value of $\tilde{D}_{i} (y)$. We are interested in the minimum of ${\tilde{D}_{i} (y)}$ when $\zeta=-1$. 
	The Lambert W function $W_{\zeta}(z)$ is asymptotic to \cite{Corless1996}
	\begin{align}
		W_{\zeta}(z)=\log z +2\pi i\zeta - \log \lb \log z +2\pi i\zeta \rb + o(1).
	\end{align}
	Therefore, for $\zeta=-1$ and large $K \rightarrow \infty$, we arrive at an asymptomatic solution for $\hat{y}$, which is given by

	\begin{align}
		\hat{y}&=\frac{\alpha}{K}\Vast( \log \lb \frac{ Kc\lb d^3K\rb ^{\frac{1}{\alpha}}}{\alpha} \rb -2\pi i - 
		\log \lb \log \lb \frac{ Kc\lb d^3K\rb^{\frac{1}{\alpha}}}{\alpha} \rb -2\pi i \rb + o(1) \Vast) \\
		&=\frac{\alpha}{K}\Vast( \underbrace{ \log \lb \frac{ -Kc\lb d^3K\rb^{\frac{1}{\alpha}}}{\alpha} \rb}_{w(K)} -\log \log \lb \frac{- Kc\lb d^3K\rb^{\frac{1}{\alpha}}}{\alpha} \rb + o(1) \Vast) \label{solver0w}\\
		&=\frac{\alpha}{K} \lb w\lb K \rb  - o\lb w\lb K \rb \rb + o(1) \rb\label{solver1w}\\
		&=\frac{1}{K} \Bigg(  \underbrace{\lb\alpha+1\rb}_{A} \log K+  \underbrace{ \log {\lb d^3c^{\alpha}\rb}  -\alpha\log \lb-\alpha\rb - \alpha o\lb w(K) \rb +\alpha o(1)}_{B} \Bigg) \\
		&= \frac{1}{K} \lb A \log K + B\rb \label{upperAB}
	\end{align} 
	where $w(K)= \log \lb \frac{ -Kc\lb d^3K\rb^{\frac{1}{\alpha}}}{\alpha} \rb$, $A=\alpha+1$ and $B={ \log {\lb d^3c^{\alpha}\rb}  -\alpha\log \lb-\alpha\rb - \alpha o\lb w(K) \rb +\alpha o(1)}$. Equality (\ref{solver0w}) is obtained due to natural logarithm function of a negative value $m<0$ is $\log m=\log(-m)+2\pi i$. Equality (\ref{solver1w}) follows from the fact $\lim_{K \rightarrow \infty}=\frac{\log (w(K))}{w(K)}=0$. Therefore, the corresponding choice of a threshold that minimizes ${\tilde{D}_{i} (y)}$ can be calculated as
	\begin{align}
		\xtrehat&=\lb\frac{\hat{y}}{c}\rb^{\frac{1}{d^2}} \nonumber\\ 
		&=\lb \frac{ A \log K + B}{cK} \rb ^{\frac{1}{d^2}}.
	\end{align} 
	Using this results, we arrive at the following lemma, which will be used for the calculation of the achievable DoF.
	
	\begin{lemma} \label{eqpairlemma2}
		If we choose the threshold $\xtrehat$ such that $\hat{y}=\frac{1}{K} \lb A \log K + B\rb$, the upper bound $\tilde{D}_i(\hat{y})$ in (\ref{ublemmaexp1}) is asymptotically equivalent to $(\frac{A\log K}{cK})^{\frac{1}{d^2}}$ when the number of users $K \rightarrow \infty$, such that 
		\begin{align}
			\lim_{K \rightarrow \infty} \frac{\tilde{D}_i(\hat{y})}{(\frac{A\log K}{cK})^{\frac{1}{d^2}}}=1 \label{lemmaforalld}.
		\end{align}
	\end{lemma}
	
	\begin{proof}
		Plugging (\ref{upperAB}) into the left hand side of (\ref{lemmaforalld}), we have
		\begin{align}
			&\lim_{K \rightarrow \infty} \frac{  (\frac{\hat{y}}{c})^{\frac{1}{d^2}}+d e^{-K\hat{y}}}{(\frac{A\log K}{cK})^{\frac{1}{d^2}}} \\
			&=\lim_{K \rightarrow \infty} \frac{ (\frac{A\log K+B}{cK})^{\frac{1}{d^2}}}{(\frac{A\log K}{cK})^{\frac{1}{d^2}}}+ \lim_{K \rightarrow \infty} \frac{d{e^{-B}}{K^{\frac{1}{d^2}-A}}}{(\frac{A\log K}{c})^{\frac{1}{d^2}}}  \label{limtbound}\\
			&=1. \nonumber
		\end{align}
		The second term of (\ref{limtbound}) equals to zero due to $\frac{1}{d^2}-A =0$, so the numerator is a constant and the denominator goes to infinity. 
		Thus, the proof is complete.
	\end{proof}
	
	
		\begin{theorem}\label{themdl1} 
			If the number of users is scaled as $K \propto P^{dd'}$, the feedback of only 1-bit per user is able to achieve the DoF $d' \in [0,d]$ per transmitter if the threshold $\xtrehat$ is chosen such that
			\begin{align}
			c{\xtrehat}^{d^2}=\frac{1}{K} \lb A \log K + B\rb.
			\end{align}
		\end{theorem}
	
	\begin{proof}
		The proof is similar to the proof of Theorem \ref{themdeq1}. The achievable DoF of transmitter $i$ using OIA can be expressed as $d-d_{\rm loss}$. If $K \propto P^{dd'}$, the DoF loss term can be written as 
		\begin{align}
			d_{\rm loss}&=d \cdot \lim_{P\rightarrow\infty} \frac{{\mathbb E}[{R_{\rm loss}}_{i}^{k^{\dagger}}]}{\log_2 P} \nonumber\\
			& \leq d \cdot \lim_{P\rightarrow\infty} \frac{{\log}_2 \left(1+\frac{P}{d}\tilde{D}_i(\hat{y})\right)}{\log_2 P} \label{dofubdl1}\\
			& = d \cdot \lim_{P\rightarrow\infty} \frac{{\log}_2 \left(1+\frac{P}{d} \lb\frac{A\log K}{cK} \rb^{\frac{1}{d^2}}\right)}{\log_2 P} \label{dofubdl2}\\
			&= d \cdot \lim_{P\rightarrow\infty} \frac{{\log}_2 \left(\frac{P}{dK^{\frac{1}{d^2}}}\right)+{\frac{1}{d^2}}{\log}_2 \lb\frac{A \log K}{c} \rb}{{\log}_2 P} \label{eqpairdl1}\\
			&=(d-d')+\lim_{P\rightarrow\infty} \frac{1}{{\log} P + O(1)} \label{eqpair1dl1}\\
			&=(d-d').
		\end{align}
		The inequality (\ref{dofubdl1}) is obtained by using the upper bound of (\ref{ublemmaexp1}). Equality (\ref{dofubdl2}) follows from the asymptotic equivalence proved in Lemma \ref{eqpairlemma2}. Equality (\ref{eqpair1dl1}) is obtained using the relationship $K \propto P^{dd'}$ and the L'Hôpital's rule. Therefore, DoF $d'$ can be achieved at each transmitter.
	\end{proof}
	
	\begin{remark} \label{solremark}
		The achieved DoF is independent of the specific value of $B$. Therefore, theorem \ref{themdl1} is valid for all $B \in \mathbb{R}$. 
		For $d=1$, the optimal threshold obtained in (\ref{opttre}) is a special case of the above result $\xtrehat=\hat{y}=\frac{1}{K} \lb A \log K + B\rb $ when $A=1$. The asymptotic equivalence can be shown as follows
		\begin{align}
			\lim_{K \rightarrow \infty} \frac{\frac{1}{K} \lb  \log K + B\rb}{1-\lb\frac{1}{K}\rb^{\frac{1}{K-1}}}
			&=\lim_{M \rightarrow 0} \frac{-M \log M}{1-M^M} \label{solutioneq1}\\
			&=\lim_{M \rightarrow 0} \frac{1}{M^M} \label{solutioneq2}\\
			&=1 \nonumber
		\end{align}
		where $M=\frac{1}{K}$ replaces $K$ for simplicity. Equality (\ref{solutioneq2}) follows from the L'Hôpital's rule.
	\end{remark}
	
	\begin{theorem} \label{corollary1}
		When the transmit power is a finite value and the number of users tends to infinity i.e.~ $P=O(1)$ and $K\rightarrow\infty$, OIA with 1-bit feedback and OIA with perfect real-valued feedback achieve the same rate.
	\end{theorem}
	\begin{proof}		
When $P=O(1)$ and $K\rightarrow\infty$, the achievable rate of OIA with perfect real-valued feedback becomes the ergodic capacity of the $d \times d$ point-to-point MIMO system without interference \cite{Lee2013b}. To complete our proof, we just need to show that OIA with 1-bit feedback achieves the same ergodic capacity of the $d \times d$ point-to-point MIMO system without interference. Therefore, we proof as follows.

When $K\rightarrow\infty$, the rate loss in (\ref{rloss}) can be written as 
\begin{align}
	{\mathbb E}[{R_{\rm loss}}_{i}^{k^{\dagger}}] \leq&  d \cdot{\rm{log}}_2 \left(1+\frac{P}{d} \cdot {{\tilde{D}_{i} (y)}} \right)
\end{align}
using the upper bound obtained in (\ref{ublemmaexp1}). If we choose the threshold $\xtrehat$ such that $\hat{y}=\frac{1}{K} \lb A \log K + B\rb$, we have
\begin{align}
\lim_{K \rightarrow \infty} \tilde{D}_i(\hat{y}) &= \lim_{K \rightarrow \infty} \left(\frac{A\log K}{cK}\right)^{\frac{1}{d^2}} \label{corollary11}\\
&=\lb \lim_{K \rightarrow \infty} \frac{A}{cK}  \rb ^{\frac{1}{d^2}} \label{corollary12}\\
&=0 \nonumber
\end{align}
where (\ref{corollary11}) follows from lemma~\ref{eqpairlemma2} and (\ref{corollary12}) is due to L'Hôpital's rule. Correspondingly, the rate loss term ${\mathbb E}[{R_{\rm loss}}_{i}^{k^{\dagger}}]$ goes to zero due to finite $P$. 
Therefore, when the number of users $K\rightarrow\infty$, we can see from (\ref{rategainloss}) that OIA with 1-bit feedback achieves the interference-free rate at the selected user, i.e.~
\begin{align}
{\mathbb E}[R_{i}^{\kdragger}]&= {\mathbb E}\Bigg[{\rm{log}}_2 {\rm{det}} \Big( \mathbf{I}+\underbrace{ {\mathbf{U}_i^{\kdragger}}^{\rm{H}} \bfH_{i,i}^{\kdragger} }_{\bar{\bfH}_{i,i}^{\kdragger}}
\underbrace{ {{\bfH}_{i,i}^{\kdragger}}^{\rmH}{\mathbf{U}_i^{\kdragger}}}_{{\bar{\bfH}}_{i,i}^{\kdragger\rmH}} \Big) \Bigg] \label{ergorate} 
\end{align}
where $\bar{\mathbf{H}}_{i,i}^{\kdragger}={\mathbf{U}_i^{\kdragger}}^{\rm{H}} \bfH_{i,i}^{\kdragger}$ is a ${d \times d}$ matrix. Every element of $\bar{\mathbf{H}}_{i,i}^{\kdragger}$ is an i.i.d. symmetric complex Gaussian random variable with zero mean and unit variance. This is due to the fact that the $\NR \times d$ truncated unitary matrix $\mathbf{U}_i^{\kdragger}$ is independent on ${\mathbf{H}}_{i,i}^{\kdragger}$. Therefore, the rate achieved in (\ref{ergorate}) becomes the ergodic capacity of the $d \times d$ point-to-point MIMO system. This also completes our proof.
\end{proof}

\section{Simulation Results}  \label{sec5}
In this section, we provide numerical results of the sum rate and the threshold choices of OIA using 1-bit feedback. 

Fig.~\ref{sumrate_OIA-1bit} shows the achievable sum rate versus SNR of  OIA with perfect real-valued feedback and OIA with 1-bit feedback, for $N_{\rm{R}}=2$, $d=1$ and the number of users $K = \lceil P \rceil$. We include also the sum rate achieved by closed-form IA in a 3-user $2\times 2$ MIMO interference channel. The threshold of our feedback scheme is calculated according to (\ref{opttre}). We can see that OIA with 1-bit feedback achieves a slightly lower rate than OIA with perfect feedback. At $30\,$dB SNR, it can achieve $90\%$ of the sum rate obtained by perfect feedback OIA. Importantly, OIA with 1-bit feedback is able to capture the slope and achieve the DoF $d=1$ (see the reference line in Fig.~\ref{sumrate_OIA-1bit}). 

The feedback mechanism can be designed in a way where any user whose distance measure is above the prescribed threshold will stay silent, and only eligible users will attempt to feedback \cite{Tang2005}. In such a mechanism, since only the eligible users feed back information, the feedback must consist of user identity and be performed on a shared random access channel, e.g., using a contention-based approach \cite{Tang2005}. It should be noted that any feedback information cannot be decoded when more than two users collide simultaneously using the same feedback resource. Therefore, the number of users that compete for the same feedback resource will have an impact on the successful transmission of the feedback information. We can establish the average number of eligible users as follows 
\begin{align}
N_{\rm bits}=KF_\mcD\left(\xtre\right).
\end{align}
Fig.~\ref{fbbits_OIA-1bit} also shows the number of eligible users per cell when the total number of users $K = \lceil P \rceil$. It can be seen that the average number of eligible users is almost a linear function with SNR (in dB) and the average number of eligible users at $30\,$dB is less than $1\%$ of the total number of users. Therefore, the small number of eligible users may ease the design of a contention-based feedback protocol.

Fig.~\ref{asm_tre_deq1} compares the threshold as a function of the number of users $K$ for $N_{\rm{R}}=4$, $d=2$. The thresholds are obtained by numerical minimization of (\ref{upeq41}), (\ref{solutionlambert}) with $\zeta=-1$ and the asymptotic expression  $\frac{A\log K}{K}$ as mentioned in {\em Remark} \ref{solremark}. The thresholds obtained by the numerical approach and by (\ref{solutionlambert}) are very close, even for a small number of users $K$. The asymptotic threshold $\frac{A\log K}{K}$ is smaller than the others since we neglect $B$ in (\ref{upperAB}). However, $B$ has no impact on the achieved DoF as explained in {\em Remark} \ref{solremark}. It can be seen that these thresholds are getting closer to each other as $K$ increases. These results validate the calculation of the thresholds.

\begin{figure}[t]
	\centering
		\centering 
		\includegraphics[width=.6\columnwidth]{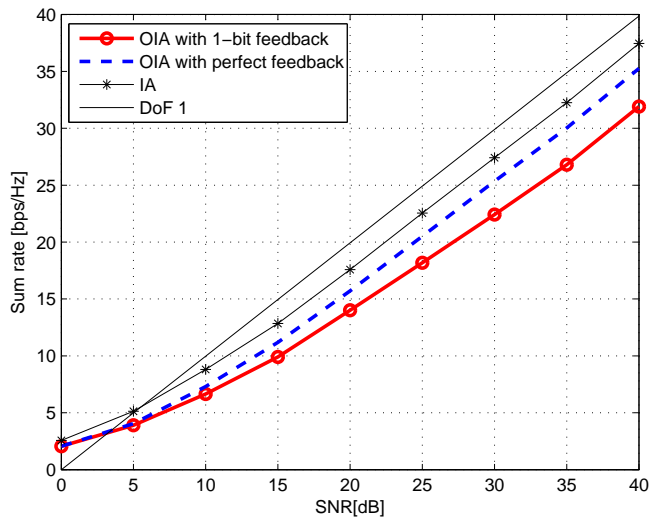}
		\caption{Achievable sum rate for $N_{\rm{R}}=2$, $d=1$. The number of users $K = \lceil P \rceil$ for OIA.}
		\label{sumrate_OIA-1bit}
\end{figure}
\begin{figure}[t]
		\centering 
		\includegraphics[width=.6\columnwidth]{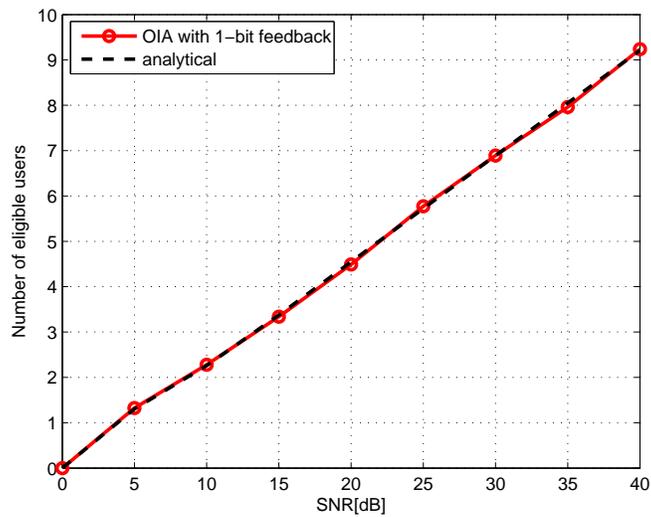}
		\caption{The average number of eligible users for $N_{\rm{R}}=2$, $d=1$ and $K = \lceil P \rceil$.}
		\label{fbbits_OIA-1bit}
\end{figure}

\begin{figure}[t]
	\centering
	\includegraphics[width=.6\columnwidth]	{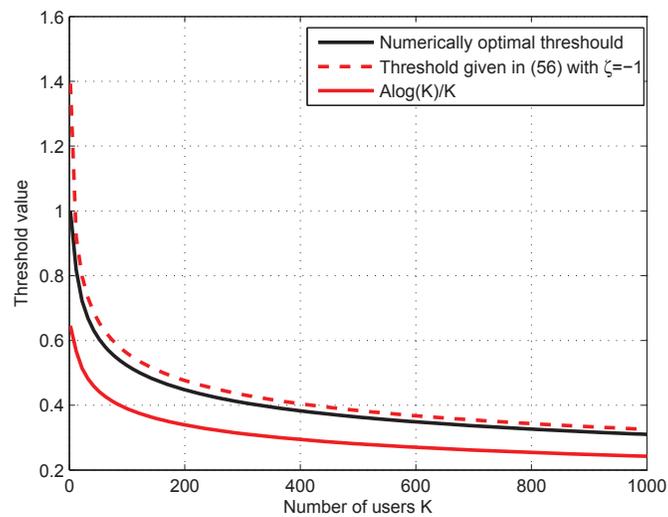}
	\caption{Comparison of the threshold obtained by numerical minimization of (\ref{upeq41}), (\ref{solutionlambert}) and the asymptotic solution $\frac{A\log K}{K}$ for $\NR=4, d=2$.}
	\label{asm_tre_deq1}
\end{figure}

Fig.~\ref{sumrate_OIA-d42} presents the sum rate versus SNR of OIA with perfect feedback and OIA with 1-bit feedback, for $N_{\rm{R}}=4$, $d=2$ and the number of users $K \in \{10,50,100\}$. The number of users does not scale with SNR, thus the sum rates saturate as SNR increases. With the increase of number of users, a higher rate is achieved. Importantly, 1-bit feedback promises about $90\%$ of the rate achieved by OIA with perfect feedback.

\begin{figure}[t]
	\centering
	\centering 
	\includegraphics[width=.6\columnwidth]{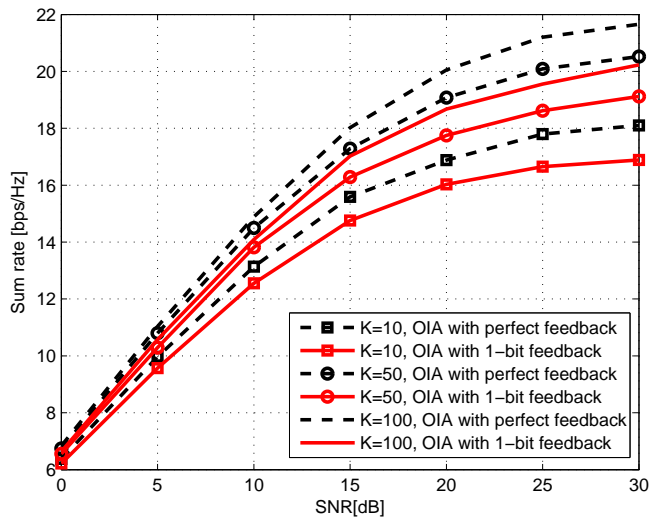}
	\caption{Achievable sum rate for $N_{\rm{R}}=4$, $d=2$. The number of users $K \in \{10,50,100\}$ (for the curves from bottom up).}
	\label{sumrate_OIA-d42}
\end{figure}

\section{Comparison OIA and IA with Limited Feedback } \label{sec6}
OIA achieves interference alignment by proper user selection. With the help of our proposed 1-bit quantizer, each user feeds back just 1 bit. Therefore, the relationship between the number of users and the amount of feedback can be established. On the other hand, IA requires CSI feedback at the transmitters to align the interference signals. The CSI is usually obtained by channel quantization on the Grassmannian manifold, where the index of the selected codeword is fed back to the transmitters. Due to the fact that the capacity of the feedback channel is usually very limited, it would be interesting to have a comparison of OIA and IA using the same amount of feedback. The work in \cite{Leithon2012} partially addressed this issue and compared the performance OIA and limited feedback IA. However, a comparison under the same amount of feedback has not been done since no limited feedback scheme was proposed  by prior works for OIA to the best of our knowledge. In this section, we will present the comparison in terms of complexity and achievable rate.

\subsection{IA with Limited Feedback} 
In this section, we review the IA limited feedback scheme proposed in \cite{Krishnamachari2013}. According to \cite{Krishnamachari2013}, receiver $i$ forms and feeds back an aggregated channel matrix $\bfW_i \in \mathbb{C}^{N_{\rm R} N_{\rm T} \times 2}$ as
\begin{align}
\bfW_i = \left[\bfw_{i,1}, \bfw_{i,2}\right].
\end{align}
The unit-norm vectors $\bfw_{i,1}, \bfw_{i,2} \in \mathbb{C}^{N_{\rm R} N_{\rm T} \times 1} $ are obtained by vectorizing the elements of matrices $\bfH_{i,p}$ and $\bfH_{i,q}$, i.e.~
\begin{align}
\bfw_{i,1}=\frac{\rmvec\lb\bfH_{i,p}\rb}{\left\| \rmvec\lb\bfH_{i,p}\rb \right\| }, \quad \bfw_{i,2}=\frac{\rmvec\lb\bfH_{i,q}\rb}{\left\| \rmvec\lb\bfH_{i,q}\rb \right\| }
\end{align}
where $p=(i+1 \mod 3)$ and $q=(i+2 \mod 3)$ are the indices of two interfering transmitters. 
Using the concept of composite Grassmannian manifold, the matrix $\bfW_i$ can be quantized using a codebook $\mathcal{C}$ with $2^{N_{\rm bits}}$ codewords and $N_{\rm bits}$ is the number of feedback bits. Each codeword $\mathbf{C}_j=[\bfc_{j,1},\bfc_{j,2}] \in \mathcal{C}$ is a $N_{\rm R} N_{\rm T} \times 2$ matrix with $\|\bfc_{j,1}\|=\|\bfc_{j,2}\|=1$. The squared distance between $\mathbf{C}_j$ and $\bfW_i$ is defined as 
\begin{align}
d_{\rm s}\lb\bfW_{i},\bfC_{j}\rb=d_{\rm c}^2\lb\bfw_{i,1},\bfc_{j,1}\rb+d_{\rm c}^2\lb\bfw_{i,2},\bfc_{j,2}\rb,
\end{align}
which is a commonly used distance measure on the composite Grassmannian manifold. The receiver $i$ calculates the squared distance $d_{\rm s}$ between $\bfW_{i}$ and every codeword in the codebook $\mathcal{C}$ and feeds back the index of the codeword which minimizes the squared distance. Based on the feedback indices from the receiver, the transmitters can obtain the quantized version of channel matrices $\bfH_{i,j},\forall i\neq j$. Then, IA precoders and decoders can be calculated according to the quantized channel matrices.

\subsection{Complexity Analysis} 
In this section, we quantify and compare the computational complexity of OIA and IA in terms of number of floating point operations (FLOPs). We will pay particular attention to the quantization process. 

One FLOP is one floating point operation, which corresponds to a real addition, multiplication, or division \cite{Golub1996}. A complex addition and multiplication require 2 FLOPs and 6 FLOPs, respectively. For a complex-valued matrix $\mathbf{A} \in \mathbb{C} ^{M \times N}$ $(M\geq N)$, the FLOP counts, denoted by $\Xi$, of some basic matrix operations are given as follows. 

\begin{itemize}
	\item Frobenius norm of $\|\bfA\|_{\rm F}$:  $\Xi_{\rm F}(M,N)=4MN$ 
	\item Gram-Schmidt orthogonalization (GSO) of $\bfA$: $\Xi_{\rm GSO}(M,N)=8N^2M-2MN$
	\item Matrix multiplication of $\bfA\bfA^{\rmH}$: $\Xi_{\otimes}(M,N)=8N^2M-2MN$
\end{itemize}

For OIA, each user needs to calculate the chordal distance between two $\NR \times d$ interference channels. According to (\ref{defdc}), the calculation of the chordal distance requires two GSOs to calculate the orthonormal bases of the two interference channels, two matrix multiplications of the truncated unitary matrices, a matrix addition of two truncated unitary matrices and a Frobenius norm operation. We ignore the scalar operations. Therefore, the total FLOPs per cell are counted as
\begin{align}
&\Xi_{\rm OIA-1bit} \nonumber\\
&=N_{\rm bits}( 2\Xi_{\rm GSO}(\NR,d)+2\Xi_{\otimes}(\NR,d)+2\NR d+\Xi_{\rm F}(\NR,d)) \nonumber\\
&=N_{\rm bits}(32{N_{\rm{R}}}d^{2} - 2\NR d).
\end{align}
where $N_{\rm bits}=K$ is the number of feedback bits since each user feeds back 1 bit. 

For IA with limited feedback, the squared distance is used for the selection of the quantized channel matrix. Thus, $2^B$ squared distance calculations will be performed in order to find the codeword. The squared distance calculates twice the chordal distance between two $\NR\NT \times 1$ vectors. Therefore, the total FLOP counts are given by
\begin{align}
\Xi_{\rm IA-joint}=2^{N_{\rm bits}}(64\NR\NT-4\NR\NT).
\end{align}

Since the joint quantization over the composite Grassmannian manifold yields a high complexity for decoding, then the quantizations of $\bfw_{i,1}$ and $\bfw_{i,1}$ over individual Grassmannian manifold $\mathcal{G}_{\NR\NT,1}(\mathbb{C})$ could be used to reduce the complexity at the expense of lower quantization resolution. Assuming equal division of the total $B$ quantization bits, the total FLOP counts of individual quantization are given by
\begin{align}
	\Xi_{\rm IA-indv}=2^{\frac{N_{\rm bits}}{2}}(64\NR\NT-4\NR\NT).
\end{align}

\begin{figure}[t]
	\centering 
	\includegraphics[width=.6\columnwidth]{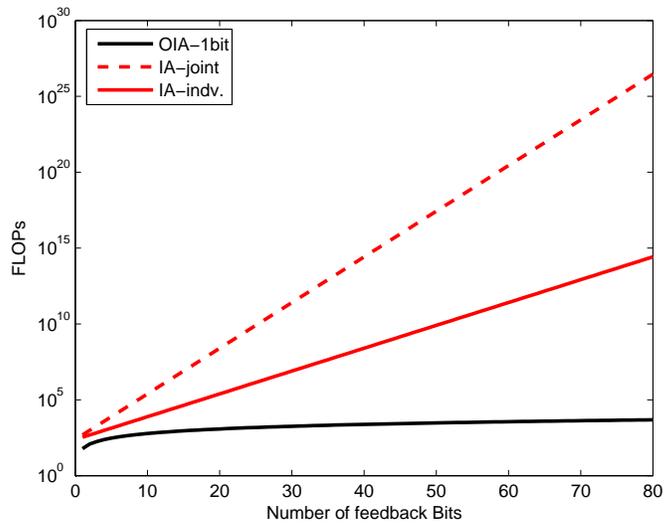}
	\caption{Feedback Complexity per cell of OIA and IA for $\NR=2$, $d=1$ ($\NT=2$ for IA). The FLOP counts of OIA sum over all $N_{\rm bits}=K$ users in a cell. }
	\label{Complexity}
\end{figure}

\begin{figure}[t]
	\centering 
	\includegraphics[width=.6\columnwidth]{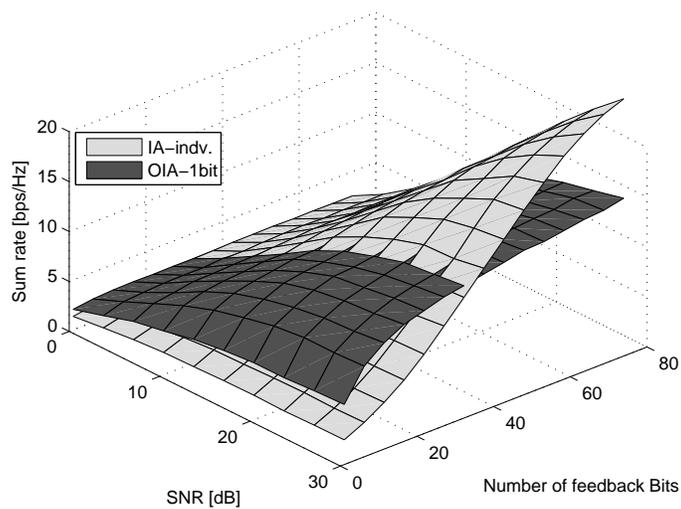}
	\caption{Sum rate for $N_{\rm{R}}=2$, $d=1$ at varying SNRs and the numbers of feedback bits per cell.}
	\label{OIA&IA_3D}
\end{figure}

The computational complexity of OIA and IA versus the number of feedback bits is given in Fig.~\ref{Complexity}. The codebook for IA with joint quantization contains $2^{N_{\rm bits}}$ codewords, which results in an exponentially increased FLOP counts. Individual quantization reduces the exponent to $\frac{N_{\rm bits}}{2}$. On the contrary, the complexity of OIA increases linearly with $N_{\rm bits}$.

Fig.~\ref{OIA&IA_3D} presents the sum rate of OIA with 1-bit feedback and IA with individual quantization. To satisfy the feasibility condition, we choose $\NT=2$ for IA. The codewords for IA are generated through random vector quantization (RVQ). In order to enable the performance analysis with exponentially growing codebook, we replace the RVQ process by a statistical model of the quantization error using random perturbations \cite[Sec.~VI.B]{Rezaee2012a}, which has shown to be a good approximation of the quantization error using RVQ. It can be observed that OIA outperforms IA when the amount of feedback is lower than 30 bits and the rate difference increases with SNR. This is due to the fact that the IA algorithm is highly sensitive to the imperfection of CSI, thus leading to a significant rate loss. At $20$~dB SNR with 10 feedback bits per cell, it can be observed that OIA compared to IA increases the sum rate by 100\% while reducing the computational complexity by more than one order of magnitude. When the number of feedback bits is larger than 30, IA starts to outperform taking advantage of the accurate CSI provided by the exponentially increased codebook size. However, the performance improvement of IA also comes with an exponentially increased computational complexity and storage, which poses a strong practical limit. From an implementation point of view, OIA with 1-bit feedback provides a better performance in the favorable operation region and enjoys a much lower complexity.

\section{Conclusion} \label{sec7}
In this paper, we analyzed the achievable DoF using a 1-bit quantizer for OIA. We proved that 1-bit feedback is sufficient to achieve the optimal DoF of $d$ in 3-cell MIMO interference channels. Most importantly, the required user scaling law remains the same as for OIA with perfect real-valued feedback. We derived a closed-form threshold for $d=1$. In the case of $d >1$, an asymptotic threshold choice was given, which is optimal when the number of users $K \rightarrow \infty$. We compared OIA and IA with the same amount of feedback and present the comparison in terms of complexity and achievable rate. At 20dB SNR with 10 bits feedback per cell for both, OIA and IA, we demonstrated that OIA reduces the complexity by more than one order of magnitude while increasing the sum rate by a factor of 2.

\bibliographystyle{IEEEtran}
\begin{footnotesize}
\bibliography{IEEEabrv,E:/Dropbox/ZhinanBIB/library}
\end{footnotesize}
	
\end{document}